\RequirePackage{amsmath}
\documentclass[a4paper]{llncs}

\usepackage{amssymb}
\usepackage{enumerate}
\usepackage{tikz}
\usepackage{algorithm}
\usepackage{algorithmicx}
\usepackage[noend]{algpseudocode}

\usepackage{hyperref}
\hypersetup{
  pdftitle = {Optimal one-shot quantum algorithm for EQUALITY and AND},
  pdfauthor = {Janis Iraids},
  pdfkeywords = {},
  pdfstartview=FitH,
  unicode=true
}

\let\doendproof\endproof
\renewcommand\endproof{~\hfill\qed\doendproof}

\makeatletter
\let\OldStatex\Statex
\renewcommand{\Statex}[1][3]{%
  \setlength\@tempdima{\algorithmicindent}%
  \OldStatex\hskip\dimexpr#1\@tempdima\relax}
\makeatother

\newcommand{\ket}[1]{\left| #1 \right\rangle}

\newcommand{\err}[1]{\mathcal{E}\left( #1 \right)}
\newcommand{\errc}[1]{\mathcal{E}^C\left( #1 \right)}
\newcommand{\errt}[2]{\mathcal{E}_{#1}\left( #2 \right)}
\newcommand{\comment}[1]{}
\newcommand{\f}[1]{\textsc{#1}}
\def\H{{\cal H}}
\def\A{{\cal A}}

\title{Optimal One-shot Quantum Algorithm for \f{EQUALITY} and \f{AND}}
\titlerunning{Optimal One-shot Quantum Algorithm for \f{EQUALITY} and \f{AND}}
\author{Andris Ambainis \and J\=anis Iraids}
\authorrunning{Ambainis \and Iraids}

\institute{Faculty of Computing, University of Latvia, Rai\c{n}a bulv\=aris 19, Riga, LV-1586, Latvia,\email{ambainis@lu.lv},\email{janis.iraids@gmail.com}}

\begin{document}
\maketitle

\begin{abstract}
  We study the computation complexity of Boolean functions in the quantum black box model. In this model our task is to compute a function $f:\{0,1\}\to\{0,1\}$ on an input $x\in\{0,1\}^n$ that can be accessed by querying the black box. Quantum algorithms are inherently probabilistic; we are interested in the lowest possible probability that the algorithm outputs incorrect answer (the error probability) for a fixed number of queries. We show that the lowest possible error probability for $\f{AND}_n$ and $\f{EQUALITY}_{n+1}$ is $\frac{1}{2}-\frac{n}{n^2+1}$.
\end{abstract}

\section{Introduction}

In this paper we study the computational complexity of Boolean functions in the quantum black box model. It is a generalization of the decision tree model, where we are computing an $n$-bit function $f:\{0,1\}^n\to \{0,1\}$ on an input $x\in\{0,1\}^n$ that can only be accessed through a black box by querying some bit $x_i$ of the input. 
In the quantum black box model the state of the computation is described by a quantum state from the Hilbert space $\H_Q\otimes \H_W \otimes \H_O$ where $\H_Q=\{\ket{0},\ket{1}, \ldots, \ket{n}\}$ is the query subspace, $\H_W$ is the working memory and $\H_O=\{\ket{0},\ket{1}\}$ is the output subspace. A computation using $t$ queries consists of a sequence of unitary transformations $U_t\cdot O_x \cdot U_{t-1} \cdot O_x \cdot \ldots \cdot O_x \cdot U_0$ followed by a measurement, where the $U_i$'s are independent of the input and $O_x=O_{Q,x}\otimes I\otimes I$ with
\[O_{Q,x}\ket{i}=\begin{cases}(-1)^{x_i}\ket{i}=\hat{x}_i \ket{i}\text{, if }i\in[n],\\
\ket{0}\text{, if }i=0,
\end{cases}
\]
is the query transformation, where $x_i \in \{0,1\}$ or equivalently, $\hat{x}_i \in \{-1,1\}$.
The final measurement is a complete projective measurement in the computational basis and the output of the algorithm is the result of the last register, $\H_O$. For and $0\leq \epsilon < \frac{1}{2}$ we denote by $Q_\epsilon(f)$ the smallest number of queries for an quantum algorithm outputting $f(x)$ with probability at least $1-\epsilon$. Usually the $\epsilon$ is omitted from $Q_\epsilon(f)$ because it changes $Q_\epsilon(f)$ by a constant factor, and $Q(f)$ is called the bounded error quantum query complexity of $f$. This complexity measure is widely studied as most computational problems can be expressed in the query model. The most well known examples are by \cite{Gro96,Sho97}. For the searching problem Grover's algorithm is exactly optimal as shown by \cite{Zal99}.

However, if one is interested in computing functions with constant number of inputs (for example, as a part of small circuit), then it may be useful to fix the number queries and minimize the probability of an incorrect answer. In this paper we will be concerned with quantum algorithms performing at most 1 query, thus we introduce $\err{f}$.

\begin{definition}
  Let $f$ be a Boolean function. Then let $\err{f}$ be the minimum error probability for a quantum algorithm that calculates $f$ using just one query, i.e.,
  \[\err{f} = \min_{\A:\A\text{ performs 1 query}} \max_x \Pr[\text{algorithm } \A \text{ does not output }f(x)].\]
\end{definition}

We will be focusing on two Boolean functions defined as follows:
\[\f{EQUALITY}_n(x)=\begin{cases}1\text{, if }x_1=x_2=\ldots=x_n\\
0\text{, otherwise}
\end{cases}\]
and
\[\f{AND}_n(x)=\begin{cases}1\text{, if }x_1=x_2=\ldots=x_n=1\\
0\text{, otherwise}
\end{cases}.\]
In her doctoral thesis \cite{Mis12} gave quantum algorithms showing that
\[\err{\f{EQUALITY}_3}\leq\frac{1}{10};\err{\f{AND}_2}\leq\frac{1}{10};\]
\[\err{\f{EQUALITY}_4}\leq \frac{1}{4};\err{\f{AND}_3}\leq \frac{1}{4};\]
\[\err{\f{EQUALITY}_6}\leq \frac{7}{16};\err{\f{AND}_5}\leq \frac{7}{16}.\]

Our main result asserts that
\begin{theorem}
\[ \err{\f{AND}_n} = \err{\f{EQUALITY}_{n+1}} = \frac{1}{2}-\frac{n}{n^2+1}.\]
\end{theorem}

The proof can be summarized in a series of three inequalities:
\[\frac{1}{2}-\frac{n}{n^2+1} \leq \err{\f{AND}_n} \leq \err{\f{EQUALITY}_{n+1}} \leq \frac{1}{2}-\frac{n}{n^2+1}.\]

The first inequality can be proven using a characterization of symmetric sum-of-squares polynomials known as the Blekherman's theorem.

{
  \renewcommand{\thetheorem}{\ref{thm:blek}}

\begin{theorem}[Blekherman]
Let $q(\hat x)$ be the symmetrization of a polynomial $p^2(\hat x)$ 
where $p(\hat x)$ is a multilinear polynomial of degree $t\leq \frac{n}{2}$ and $\hat x = (x_1, \ldots, x_n)$.
Then, over the Boolean hypercube $\hat{x}\in\{-1,1\}^n$,
\[ q(\hat x) = \sum_{j=0}^t p_{t-j}(|x|) \left( \prod_{0\leq i < j} (|x|-i) (n-|x|-i) \right) \]
where $p_{t-j}$ is a univariate polynomial that is a sum of squares of polynomials of degree at most $t-j$
and $|x|$ denotes the number of variables $i:\hat x_i=-1$.
\end{theorem}
\addtocounter{theorem}{-1}
}
Even though it is an unpublished result, there are proofs --- see \cite{LP+16} or Section~\ref{sec:blek} in this paper for a considerably shorter proof using representation theory. 

The second inequality is trivial, since 
\[\f{AND}_n(x_1, \ldots, x_n)=\f{EQUALITY}_{n+1}(x_1, \ldots, x_n,1),\]
and so we can use an algorithm for $\f{EQUALITY}_{n+1}$ to calculate $\f{AND}_n$.

The third inequality can be proved by constructing a quantum algorithm for the function $\f{EQUALITY}_{n+1}$. Since the algorithm is very simple we present it before the more involved proof of the first inequality.

If we compare $\err{f}$ with the classical analogue, let us call it $\errc{f}$, \cite{Mis12} has shown that $\errc{\f{EQUALITY}_n}=\frac{1}{2}$ and $\errc{\f{AND}_n}=\frac{1}{2}-\frac{1}{4n-2}$. 

\section{Algorithm for \f{EQUALITY}}
\label{sec:ub}

\begin{theorem}
\[\err{\f{EQUALITY}_{n+1}} \leq \frac{1}{2}-\frac{n}{n^2+1}\]
\end{theorem}
\begin{proof}
We will prove that the following algorithm has the claimed error probability:
\begin{algorithm}[H]
  \caption{Algorithm for $\f{EQUALITY}_{n+1}$}
  \begin{algorithmic}[1]
    \State{State space: $\ket{1}, \ket{2}, \ldots, \ket{n+1}$}
    \State{Start in uniform superposition $\sum_{i=1}^{n+1}{\frac{1}{\sqrt{n+1}}\ket{i}}$}
    \State{Query: $\sum_{i=1}^{n+1}{\frac{1}{\sqrt{n+1}}\ket{i}}\xrightarrow{Q}\sum_{i=1}^{n+1}{\frac{(-1)^{x_i}}{\sqrt{n+1}}\ket{i}}$}
    \State{Perform quantum Fourier transform $F_{n+1}\ket{i}=\sum_{j=1}^{n+1}{\frac{\omega^{(i-1)(j-1)}(-1)^{x_j}}{n+1}}\ket{i}$,$\omega=e^{\frac{2\pi}{n+1}i}$:}
    \Statex{$\sum_{i=1}^{n+1}{\frac{(-1)^{x_i}}{\sqrt{n+1}}\ket{i}}\xrightarrow{F_{n+1}}\sum_{i=1}^{n+1}{\sum_{j=1}^{n+1}{\frac{\omega^{(i-1)(j-1)}(-1)^{x_j}}{n+1}}\ket{i}}$}
    \State{Perform a complete measurement}
    \If{the result is state $\ket{1}$}
      \State{With probability $\frac{1}{2}-\frac{n}{n^2+1}$ output 0; otherwise output 1}
    \Else
      \State{Output 0}
    \EndIf
  \end{algorithmic}
\end{algorithm}

First, let us consider the case when $\f{EQUALITY}_{n+1}=1$. In that case the state $\ket{1}$ will be measured with certainty and hence the probability to output the incorrect answer $0$ is $\frac{1}{2}-\frac{n}{n^2+1}$.

If on the other hand the input is such that $\f{EQUALITY}_{n+1}=0$, the algorithm has an opportunity to answer incorrectly only in the case it measures $\ket{1}$. Denote by $m:=\sum_{i=1}^{n+1}{x_i}$. The probability that the algorithm answers 1 is $\left(\frac{m}{n+1}\right)^2\cdot \left(\frac{1}{2}+\frac{n}{n^2+1}\right)$. The value of this expression is maximized when $m=\pm(n-1)$ and so the probability to answer 1 on the worst kind of input (namely the input where only one bit is different from every other bit) is
\[\left(\frac{n-1}{n+1}\right)^2 \left(\frac{1}{2}+\frac{n}{n^2+1}\right)=\left(\frac{n-1}{n+1}\right)^2\frac{(n+1)^2}{2(n^2+1)}=\frac{n^2+1-2n}{2(n^2+1)}=\frac{1}{2}-\frac{n}{n^2+1}.\]
\end{proof}

\section{Lower Bound for \f{AND}}
\label{sec:lb}

\begin{theorem}
\[\err{\f{AND}_n} \geq \frac{1}{2}-\frac{n}{n^2+1}\]
\end{theorem}

\begin{proof}
  First, we will restrict the domain of the inputs of the $\f{AND}_n$ function to bit lists with Hamming weight of $0$, $n-1$ or $n$. It turns out that this promise problem has the same optimal error probability. Consider any quantum algorithm computing $\f{AND}_n$ with error probability $\epsilon$. Following the familiar reasoning of \cite{BBC+98} we can write the probability that the algorithm outputs 1 as a sum-of-squares polynomial of degree at most 2:
  \[\Pr[\text{algorithm outputs }1]=\sum_i{p_i^2(\hat{x}_1, \ldots, \hat{x}_n)}.\]
  From Blekherman's theorem we obtain that by symmetrization there must exist a degree at most 2 univariate polynomial of the form
  \[p(s)=\sum_i{(a_is+b_i)^2}+(n-s)s\sum_j{c_j^2}\]
  such that
  \[1-\epsilon \leq p(s) \leq 1\]
  when $s=|x|=n$ and
  \[0 \leq p(s) \leq \epsilon\]
  when $s=|x|$ and $s\in\{0,n-1\}$.

  A geometric representation of the potential regions where $p(s)$ intersects $s=0$, $s=n-1$ and $s=n$ is depicted in Figure~\ref{fig:err}.
  \begin{figure}
\begin{center}
  \begin {tikzpicture}[scale=2]
    \draw (-0.1, 0) -- (4.1, 0);
    \draw[->] (0, -0.1) -- (0, 1.1);
    \draw[dashed] (0, 0.3) -- (4.1, 0.3);
    \draw[dashed] (0, 0.7) -- (4.1, 0.7);
    \draw (0.07,0) -- (-0.07, 0) node[anchor=east] {$0$};
    \draw (0.05,0.3) -- (-0.05, 0.3) node[anchor=east] {$\epsilon$};
    \draw (0.07,0.5) -- (-0.07, 0.5) node[anchor=east] {$\frac{1}{2}$};
    \draw (0.05,0.7) -- (-0.05, 0.7) node[anchor=east] {$1-\epsilon$};
    \draw (0.07,1) -- (-0.07, 1) node[anchor=east] {$1$};
    
    \draw (0,0.07) -- (0, -0.07) node[anchor=north] {$0$};
    \draw (3,0.07) -- (3, -0.07) node[anchor=north] {$n-1$};
    \draw (4,0.07) -- (4, -0.07) node[anchor=north] {$n$};

    \draw[ultra thick] (0,0) -- (0,0.3);
    \draw[ultra thick] (3,0) -- (3,0.3);
    \draw[ultra thick] (4,0.7) -- (4,1);
  \end {tikzpicture}
\end{center}
  \caption{Regions where $p(s)$ may intersect $s=0$, $s=n-1$ and $s=n$}
  \label{fig:err}
  \end{figure}
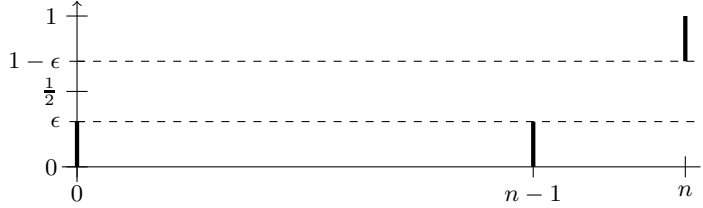
  Clearly, a degree 0 polynomial $p$ --- a constant, would produce an error probability $\epsilon=\frac{1}{2}$. Consider a degree 1 polynomial --- a straight line. We will apply transformations to the line that do not increase (but may decrease) the error probability it achieves $\epsilon$. First, we stretch the line vertically with respect to the horizontal line $y=\frac{1}{2}$ until it passes through the origin. Then we stretch the line vertically with respect to the horizontal line $y=0$ until it passes through both $(n-1,\epsilon)$ and $(n,1-\epsilon)$. This line has a slope $\frac{\epsilon}{n-1}=\frac{1-\epsilon}{n}$ and so $\epsilon=\frac{1}{2}-\frac{1}{4n-2}$.

  Finally, consider a degree 2 polynomial $p$ --- a parabola. If the parabola is concave, we may reason similarly as in the line case, except, the point $(n-1,\epsilon)$ must now be above the line passing through $(0,0)$ and $(n,1-\epsilon)$ and so the error probability is higher. If the parabola is convex, we consider further two cases.
  \begin{enumerate}[a)]
  \item If the vertex of the parabola has $s\leq0$, then we perform the same vertical stretchings. Since the parabola now passes through $(0,0)$ we can describe it with an equation $as^2+bs$ where $a>0$. Since the vertex of the parabola has $s\leq 0$, the coefficient $b$ must be non-negative. The smallest $\epsilon$ possible for such parabolas can be described through the system
    \[\begin{split}
    1-\epsilon&=\max_{a,b}{an^2+bn}\quad \text{such that}\\
    &\quad\quad\begin{cases}an^2+bn+a(n-1)^2+b(n-1)=1\\
        b\geq 0
     \end{cases}
    \end{split}
    \]
    From the equality we can express $b=\frac{1-a(n^2+(n-1)^2)}{2n-1}$ and hence $a\leq \frac{1}{n^2+(n-1)^2}$. Plugging it all into the objective function we have that
    \[\begin{split}1-\epsilon &\leq an^2+\frac{n(1-a(n^2+(n-1)^2))}{2n-1}\leq\\
    &\leq \frac{n^2}{n^2+(n-1)^2} \leq \frac{(n+1)^2}{2(n^2+1)} = \frac{1}{2}+\frac{n}{n^2+1}
    \end{split}
    \]
  \item
    If the vertex of the parabola has $s\geq 0$ then clearly the vertex has to be in the interval $s\in[0,n]$. Therefore we use the property from Blekherman's characterization that the polynomial $p(s)$ is non-negative in the interval $s\in[0,n]$, i.e., the term $\sum_i{(a_is+b_i)^2}$ is non-negative everywhere and $(n-s)s\sum_j{c_j^2}$ is non-negative for $s\in[0,n]$. Now we stretch the parabola horizontally with respect to line $s=n$ until $p(0)=p(n-1)$. This will not increase $\epsilon$ and preserve the non-negativity in the interval $s\in[0,n]$. Next we stretch the parabola vertically with respect to line $y=p(n)$ until $p\left(\frac{n-1}{2}\right)=0$. Again, this step does not increase $\epsilon$. Finally, we stretch vertically with respect to $y=0$ until $p(0)=1-p(n)$. The last step preserved the vertex at $(\frac{n-1}{2},0)$ so the parabola has an equation
    \[p(s)=a\left(s-\frac{n-1}{2}\right)^2.\]
    But from the equation $p(0)=1-p(n)$ we obtain
    \[a\left(\frac{n-1}{2}\right)^2=1-a\left(\frac{n+1}{2}\right)^2;\]
    \[a=\frac{2}{n^2+1}.\]
    Consequently,
    \[\epsilon \geq p(0) = \frac{2}{n^2+1}\left(\frac{n-1}{2}\right)^2=\frac{1}{2}-\frac{n}{n^2+1}.\]
  \end{enumerate}
\end{proof}
Interestingly, the proof only really requires the sum-of-squares characterization when $\frac{n-1}{2}$ is not an integer. The fact that the parabola $p(s)$ is non-negative at $s=\frac{n-1}{2}$ is sufficient.

\section{Proof of Blekherman's Theorem}
\label{sec:blek}

In this section we prove Blekherman's theorem.
\begin{theorem}[Blekherman]
\label{thm:blek}
Let $q(\hat x)$ be the symmetrization of a polynomial $p^2(\hat x)$ 
where $p(\hat x)$ is a multilinear polynomial of degree $t\leq \frac{n}{2}$ and $\hat x = (x_1, \ldots, x_n)$.
Then, over the Boolean hypercube $\hat{x}\in\{-1,1\}^n$,
\[ q(\hat x) = \sum_{j=0}^t p_{t-j}(|x|) \left( \prod_{0\leq i < j} (|x|-i) (n-|x|-i) \right) \]
where $p_{t-j}$ is a univariate polynomial that is a sum of squares of polynomials of degree at most $t-j$
and $|x|$ denotes the number of variables $i:\hat x_i=-1$.
\end{theorem}
Our proof utilizes concepts of representation theory. For a description of the core tools of representation theory that we require refer to the first two chapters of \cite{Ser77}.

\subsection{Group representation}

Let $H_\wp$ be a Hilbert space with basis states $\hat x_{S}$ (for all $S\subseteq [n]$) corresponding to monomials $\prod_{i\in S}\hat  x_i$. Then, the vectors in $H_\wp$ correspond to multilinear polynomials in variables $\hat x_i$.
We consider a group representation of the symmetric group $\mathfrak{S}_n$ on $H_\wp$ with transformations $U_{\pi}$ defined by 
$U_{\pi}\hat  x_{S}=\hat x_{\pi(S)}$.
The irreducible representations contained in $H_\wp$ are well known:

Let $S_m(\hat{x}_1, \ldots, \hat{x}_n)=\sum_{i_1, \ldots, i_m} \hat{x}_{i_1} \ldots \hat{x}_{i_m}$ be the $m^{\rm th}$ elementary symmetric polynomial. We use $S_0(\hat{x}_1, \ldots, \hat{x}_n)$ to
denote the constant 1.

\begin{lemma}
  \label{lem:irred}
 A subspace $H\subseteq H_\wp$ is irreducible if and only if there exist $b$ and $\alpha_{m}$ for $m= 0, 1, \ldots, n-2b$ such that
  $H$ is spanned by vectors $\overrightarrow{p}_{i_1, \ldots, j_b}$  corresponding to polynomials
$p_{i_1, \ldots, j_b}$  (for all choices of pairwise distinct $i_1, j_1, \ldots, i_b, j_b\in [n]$) where
  \[ p_{i_1, \ldots, j_b}(\hat{x}_1, \ldots, \hat{x}_n) = (\hat{x}_{i_1}-\hat{x}_{j_1}) \ldots (\hat{x}_{i_b}-\hat{x}_{j_b}) \sum_{m=0}^{n-2b} \alpha_m S_m(\hat{x}') \]
and $\hat{x}'\in\{-1, 1\}^{n-2b}$ consists of all $\hat{x}_i$ for $i\in [n]$, $i\notin\{i_1, \ldots, j_b\}$.
\end{lemma}

See \cite{Bel15} for a short proof of Lemma \ref{lem:irred}.

\subsection{Decomposition of \texorpdfstring{$q(\hat x)$}{q(x)}}

Let
\[ p(\hat x_1, \ldots,\hat  x_n) = \sum_{S: |S|\leq t} a_S \hat x_S .\]
We associate $p^2(\hat x_1, \ldots, \hat x_n)$ with the matrix $(P_{S_1, S_2})$ 
with rows and columns indexed by $S\subseteq [n], |S|\leq t$ defined by $P_{S_1, S_2}=a_{S_1} a_{S_2}$.
Let $\overrightarrow{x}$ be a column vector consisting of all $\hat x_S$ for $S:|S|\leq t$. 
Then, $p^2(\hat x_1, \ldots, \hat x_n) = \overrightarrow{x}^T P \overrightarrow{x}$.
This means that $P$ is positive semidefinite.

For a permutation $\pi\in \mathfrak{S}_n$, let $P^{\pi}$ be the matrix defined by 
\[ P^{\pi}_{S_1, S_2} = a_{\pi(S_1)} a_{\pi(S_2)} \]
and let $Q=\frac{1}{n!} \sum_{\pi\in \mathfrak{S}_n} P^{\pi}$ be the average of all $P^{\pi}$. Then, 
$q(\hat x)=\overrightarrow{x}^T Q \overrightarrow{x}$.
$Q$ is also positive semidefinite (as a linear combination of positive semidefinite matrices $P^{\pi}$ with positive coefficients).

We decompose $Q=\sum_i \lambda_i Q_i$ with $\lambda_i$ ranging over different non-zero eigenvalues and $Q_i$ being the projectors on the respective eigenspaces. Since $Q$ is positive semidefinite, we have $\lambda_i>0$ for all $i$.

We interpret transformations $U_{\pi}$ as permutation matrices defined by $(U_\pi)_{S, S'}=1$ if $S=\pi(S')$ and 
$(U_\pi)_{S, S'}=0$ otherwise.
Then, we have 
\[ U_{\pi} Q U^{\dagger}_{\pi} = \frac{1}{n!} \sum_{\tau\in \mathfrak{S}_n} U_{\pi} P^{\tau} U^{\dagger}_{\pi} =
\frac{1}{n!} \sum_{\tau\in \mathfrak{S}_n} P^{\pi \tau} = \frac{1}{n!} \sum_{\tau\in \mathfrak{S}_n} P^{\tau} = Q.\]
Since we also have
\[ U_{\pi} Q U^{\dagger}_{\pi} = \sum_i \lambda_i U_{\pi} Q_i U^{\dagger}_{\pi}, \]
we must have $Q_i= U_{\pi} Q_i U^{\dagger}_{\pi}$. This means that $Q_i$ is a projector to 
a subspace $H_i\subseteq H_\wp$ that is invariant under the action of $\mathfrak{S}_n$.
If $H_i$ is not irreducible, we can decompose it into a direct sum of irreducible subspaces
\[ H_i = H_{i, 1} \oplus H_{i, 2} \oplus \ldots \oplus H_{i, m_i} .\]
Then, we have $Q_i = \sum_{j=1}^{m_i} Q_{i, j}$ where $Q_{i, j}$ is a projector to $H_{i, j}$
and $Q=\sum_{i, j} \lambda_i Q_{i, j}$.
This means that we can decompose $q(\hat x)=\sum_{i, j}\lambda_i q_{i, j}(\hat x)$
where $q_{i, j}(\hat x)=\overrightarrow{x}^T Q_{i, j} \overrightarrow{x}$ 
and it suffices to show the theorem for one polynomial $q_{i, j}(\hat x)$ instead of
the whole sum $q(\hat x)$.

\subsection{Projector to one subspace.}
Let $H_{\wp,\ell}\subseteq H_{\wp}$ be an irreducible invariant subspace.
We claim that the projection to the subspace $H_{\wp,\ell}$ denoted by $\Pi_{\wp,\ell}$ is of the following form:
\begin{lemma}
  \[ \Pi_{\wp, \ell} = c\rho_{\wp, \ell} \text{ where } 
  \rho_{\wp, \ell} =\sum_{i_1, \ldots, j_b} \overrightarrow{p}_{i_1, \ldots, j_b} \overrightarrow{p}^T_{i_1, \ldots, j_b} 
  \]
for some constant $c$.
\end{lemma}

\proof
If we restrict to the subspace $H_{\wp, \ell}$, then $\Pi_{\wp, \ell}$ is just the identity $I$.

On the right hand side, $\rho_{\wp, \ell}$ is mapped to itself by any $U_{\pi}$ (since any $U_{\pi}$ permutes the vectors
$\overrightarrow{p}_{i_1, \ldots, j_b}$ in some way). 
Therefore, all $U_{\pi}$ also map the eigenspaces of $\rho_{\wp, \ell}$ to themselves.
This means that, if $\rho_{\wp, \ell}$ has an eigenspace $V\subset H_{\wp, \ell}$, 
then $U_{\pi}$ acting on $V$ also form a representation of $\mathfrak{S}_n$ but that would contradict
$H_{\wp, \ell}$ being an irreducible representation.
Therefore, the only eigenspace of $\rho_{\wp, \ell}$ is the entire $H_{\wp, \ell}$.
This can only happen if $\rho_{\wp, \ell}$ is $c I$ for some constant $c$.
\qed

\subsection{Final polynomial}

From the previous subsection, it follows that $q_{i, j}(\hat x)$ is a positive constant times
\[ \sum_{i_1, \ldots, j_b}  (\hat{x}_{i_1}-\hat{x}_{j_1})^2 \ldots (\hat{x}_{i_b}-\hat{x}_{j_b})^2 S^2(\hat{x}') \]
 where $S(\hat{x}')$ is a symmetric polynomial of degree at most $t-b$.
Instead of the sum, we consider the expected value of 
$(\hat{x}_{i_1}-\hat{x}_{j_1})^2 \ldots (\hat{x}_{i_b}-\hat{x}_{j_b})^2 S^2(\hat{x}')$ 
when $i_1, \ldots, j_b$ are chosen randomly. (Since the sum and the expected value differ by a constant factor,
this is sufficient.)

Terms $(\hat x_{i_k} - \hat x_{j_k})^2$ are nonzero if and only if one of $x_{i_k}$ and $x_{j_k}$ is $1$ and the other is $-1$.
Then, for $k=1$, we have
\[ \Pr\left[\{\hat x_{i_1}, \hat x_{j_1}\} =\{-1, 1\} \right] = \frac{2s(n-s)}{n(n-1)} ,\]
since there are $\frac{n(n-1)}{2}$ possible sets $\{\hat x_{i_1}, \hat x_{j_1}\}$ and $s(n-s)$ of 
them contain one $1$ and one $-1$.
For $k>1$,
\[ \Pr\left[\{\hat x_{i_k}, \hat x_{j_k}\} =\{-1, 1\} |
\{\hat x_{i_l}, \hat x_{j_l}\} =\{-1, 1\} \mbox{ for } l\in[k-1] \right] \]
\[ = \frac{2(s-k+1)(n-s-k+1)}{(n-2k+2)(n-2k+1)} ,\]
since the condition $\{\hat x_{i_l}, \hat x_{j_l}\} =\{-1, 1\}$ for $l\in [k-1]$ means that, among the remaining variables, there are
$s-k+1$ variables $\hat x_j=-1$ and $n-s-k+1$  variables $\hat x_j=1$ and $n-2k+2$ variables in total (and, given that,
the $k=1$ argument applies). Thus,
\[ \Pr \left[ \prod_{k=1}^b{(\hat{x}_{i_k}-\hat{x}_{j_k})^2} =1 \right] = \frac{2^b s(s-1) \ldots (s-b+1) (n-s) \ldots (n-s-b+1)}{n(n-1) \ldots (n-2b+1)} .\]
Since $S$ is a symmetric polynomial, we have $S(\hat{x}')=S'(s')$ where $S'$ is a polynomial of one variable $s'$, with $s'$ equal to
the number of variables $\hat{x}'_j=-1$. Since there are $b$ variables $\hat{x}_j=-1$ that do not appear in $\hat{x}'$,
we have $s'=s-b$. This means that $S'$ can be rewritten as a polynomial in $s$ (instead of $s'$).

\section{Conclusion}
In this paper we have shown that
\[ \err{\f{AND}_n} = \err{\f{EQUALITY}_{n+1}} = \frac{1}{2}-\frac{n}{n^2+1}.\]
There is a natural way to generalize $\err{f}$ to any fixed number of queries $t$. We may denote it by $\errt{t}{f}$ and have
\[\errt{t}{f} = \min_{\A:\A\text{ performs }t \text{ queries}} \max_x \Pr[\text{algorithm } \A \text{ does not output }f(x)].\]
From the numerical experiments of \cite{MJM15} it seems that the connection between $\f{EQUALITY}_{n+1}$ and $\f{AND}_n$ goes much deeper.
\begin{conjecture}
  For all positive integers $t$ and $n$:
\[  \errt{t}{\f{EQUALITY}_{n+1}}=\errt{t}{\f{AND}_n}.\]
\end{conjecture}

\section*{Acknowledgements}
The research leading to these results has received
funding from the European Union Seventh Framework Programme (FP7/2007-2013) under grant agreement n${{}^\circ}$ 600700 (QALGO), ERC Advanced Grant MQC, Latvian State Research programme NexIT project No.1.

\bibliographystyle{splncs03}

\phantomsection
\addcontentsline{toc}{chapter}{References}
\bibliography{quantum}

\end{document}